\documentclass[a4paper,USenglish,cleveref,thm-restate,nolineno]{socg-lipics-v2021}

\pdfoutput=1 
\hideLIPIcs  

\title{Nearly Instance Optimal Sparse Matrix Approximation from Matrix-Vector Products} 

\author{Christopher Musco}{New York University, United States}{cmusco@nyu.edu}{}{}
\author{Indu Ramesh}{New York University, United States}{ir914@nyu.edu}{}{}

\authorrunning{C. Musco and I. Ramesh} 

\Copyright{Christopher Musco and Indu Ramesh} 


\makeatletter
\renewcommand{\subjclassHeading}{}
\renewcommand{\@ccsdescString}{}

\renewcommand{\keywordsHeading}{}
\renewcommand{\@keywords}{}
\makeatother

\relatedversion{} 

\EventEditors{John Q. Open and Joan R. Access}
\EventNoEds{2}
\EventLongTitle{42nd Conference on Very Important Topics (CVIT 2016)}
\EventShortTitle{CVIT 2016}
\EventAcronym{CVIT}
\EventYear{2016}
\EventDate{December 24--27, 2016}
\EventLocation{Little Whinging, United Kingdom}
\EventLogo{}
\SeriesVolume{42}
\ArticleNo{23}

\usepackage{multicol}
\usepackage{bbm}
\usepackage[labelfont=bf,textfont=it,font=small]{caption}
\usepackage{algorithm}
\usepackage{algorithmicx}
\usepackage[noend]{algpseudocode}
\algrenewcommand\algorithmiccomment[1]{\hfill{\color{gray}$\triangleright$~#1}}
\algrenewcommand\algorithmicrequire{\textbf{Input:}}
\algrenewcommand\algorithmicensure{\textbf{Output:}}
\newcommand{\algrule}[1][.2pt]{\par\vskip.2\baselineskip\hrule height #1\par\vskip.2\baselineskip}
\usepackage{subcaption}
\usepackage{graphicx}

\usepackage{nicematrix}
\usepackage{dsfont}

\usepackage{xcolor}
\definecolor{c0}{HTML}{641a80}
\definecolor{c1}{HTML}{b73779}

\usepackage{hyperref}
\hypersetup{
    colorlinks,
    linkcolor=c0,
    citecolor=c1,
    urlcolor=c0
}

\newcommand{\nnz}{\mathrm{nnz}}
\newcommand{\degen}{\mathrm{degen}}
\DeclareMathOperator{\argmin}{argmin}

\newcommand{\R}{\mathbb{R}}
\newcommand{\E}{\mathbb{E}}

\newcommand{\bv}[1]{\mathbf{#1}}

\newtheorem{problem}{Problem}

  \usepackage{nth}
  \usepackage{intcalc}

  \newcommand{\cAAAI}[1]{AAAI\ Conference\ on\ Artificial (AAAI)}

\nolinenumbers

\begin{document}
\maketitle

\begin{abstract}
  A large body of work studies the problem of learning an approximation to an implicit matrix $\bv{A}\in \R^{m\times n}$ that is only accessible implicitly via matrix-vector product queries (matvec queries) of the form $\bv{x} \rightarrow \bv{A}\bv{x}$ or $\bv{x} \rightarrow \bv{A}^T\bv{x}$. Of particular interest are methods that learn a near-optimal approximation with a \emph{fixed sparsity pattern}. For example, we might want to learn a  near-optimal diagonal, banded, or arrow-head approximation to an implicit matrix $\bv{A}$. 
  
  Naturally, the number of matvec queries required to solve this problem depends on the sparsity pattern, which can be encoded as a binary matrix $\bv{S}\in \{0,1\}^{m\times n}$. The query complexity of previous algorithms scales with quantities like the total number of ones in $\bv{S}$, its maximum column/row sparsity, or the chromatic number of a its ``conflict graph''. These quantities are incomparable: for a given $\bv{S}$, parameterizing by one might yield lower query complexity than another.

  In this work, we unify and tighten these prior results by providing a nearly sharp characterization of the matvec query complexity of sparse matrix approximation. Generalizing a definition from graph algorithms, let the \emph{degeneracy},  $\degen(\bv{S})$, denote the smallest number $k$ so that, if we iteratively delete all rows and columns of $\bv{S}$ with $\leq k$ ones, we are left with an empty matrix. 
  We show that a near-optimal approximation to $\bv{A}$ with sparsity pattern $\bv{S}$ can be learned with $\tilde{O}(\degen(\bv{S}))$ matrix-vector product queries, and $\Omega(\degen(\bv{S}))$ queries are necessary, for \emph{any sparsity pattern} $\bv{S}$. Moreover, unlike prior work based on graph coloring, all of our methods run in polynomial time. 
\end{abstract}

\section{Introduction}
Matrix approximation based on matrix-vector products has emerged as a fundamental primitive in numerical linear algebra and optimization \cite{AmselAviChen:2025,AmselChenKeles:2026,BakshiClarksonWoodruff:2022,ChenKelesHalikias:2025,ClarksonWoodruff:2009,HalikiasTownsend:2023,HalkoMartinssonTropp:2011,JambulapatiLiMusco:2023,LevittMartinsson:2024,LevittMartinsson:2024a,LinLuYing:2011,LinLuYing:2009,LiuXingGuo:2021,PearceYesypenkoLevitt:2025,SchaferKatzfussOwhadi:2021}. Recently, the problem has also found applications in scientific machine learning (SciML) as a powerful yet tractable special case of \emph{operator learning} \cite{BoulleHalikiasOtto:2024,BoulleHalikiasTownsend:2023,BoulleTownsend:2023,HoopKovachkiNelsen:2023,SchaferOwhadi:2024}. 

In the basic setup, we are given access to black-box matrix vector products (abbreviated ``matvecs'') with $\bv{A}$ and $\bv{A}^T$ for some unknown matrix $\bv{A} \in \R^{m\times n}$. I.e., we can evaluate $\bv{A}\bv{x}_1, \bv{A}^T\bv{y}_1, \ldots, \bv{A}\bv{x}_q, \bv{A}^T\bv{y}_q$ for adaptively chosen vectors $\bv{x}_1, \ldots, \bv{x}_q \in \R^{n}$ and $\bv{y}_1, \ldots, \bv{y}_q \in \R^{m}$. ``Adaptively'' means that the choice of $\bv{x}_i,\bv{y}_i$ can depend on all previous queries and their results. Given the results of all queries, the goal is to find a near-optimal approximation to $\bv{A}$ from some pre-specified family $\mathcal{F} \subset \R^{m\times n}$. I.e., find $\tilde{\bv{B}}\in \mathcal{F}$ satisfying\footnote{Matrix approximation is most commonly studied with error measured by the Frobenius norm, $\|\bv{A} - \bv{B}\|_F$, and this is the norm we consider. However, it is natural to study other norms, like the spectral norm \cite{MuscoMusco:2015,SimchowitzEl-AlaouiRecht:2018}. Additionally, the goal of $(1+\epsilon)$-near optimality can often be relaxed: the problem is still interesting when the approximation factor is a constant or grows mildly with $\bv{A}$'s dimensions \cite{AmselChenKeles:2026}.}:
\begin{align}
\label{eq:goal}
\|\bv{A} - \tilde{\bv{B}}\|_F \leq (1+\epsilon)\cdot \min_{\bv{B}\in \mathcal{F}}\|\bv{A} - \bv{B}\|_F.
\end{align}

Natural examples for the approximation family, $\mathcal{F}$, include low-rank matrices, hierarchically structured matrices, butterfly matrices, and more \cite{LinLuYing:2011,LiuXingGuo:2021,HalikiasTownsend:2023}. As one example result, when $\mathcal{F}$ is the class of rank-$k$ matrices, \eqref{eq:goal} can be achieved with  $O(k/\epsilon^{1/3})$ matvecs \cite{BakshiClarksonWoodruff:2022}.

In applications, we are interested in algorithms based on matvec queries because it is often far more efficient to compute matvecs than to produce $\bv{A}$ explicitly. Common settings include when $\bv{A}$ is the Hessian matrix of a high-dimensional optimization problem, with which matvecs can be computed efficiently via autodiff \cite{Pearlmutter:1994}, or when $\bv{A}$ is an integral transform that can be applied efficiently using the Fourier transform or Fast Multipole Method \cite{LevittMartinsson:2024}. Another example arises when $\bv{A}$ is the inverse of a matrix $\bv{B}$. While producing $\bv{A}$ requires matrix-inversion, and thus $O(n^3)$ time, matvecs can often be computed by solving a linear system involving $\bv{B}$, which often takes $O(n^2)$ or less when $\bv{A}$ is well-conditioned \cite{BekasKokiopoulouSaad:2007},

\subsection{Sparse Matrix Approximation}
\label{sec:prior_work}
An important special case of the matrix approximation problem is when $\mathcal{F}$ contains all matrices with some \emph{fixed sparsity pattern}.\footnote{This problem should not be confused with the closely related but distinct \emph{sparse recovery problem}, where we wish to find a good sparse approximation to $\bv{A}$, and are free to use any sparsity pattern we wish \cite{DasarathyShahBhaskar:2015}. See \cite{AmselChenKeles:2026b} for further discussion on differences between these problems.} One example is the class of diagonal matrices \cite{BekasKokiopoulouSaad:2007,BastonNakatsukasa:2022,DharangutteMusco:2023}. Diagonal matrices are widely used to approximate Hessian matrices in optimization \cite{BlondelRoulet:2024,DauphinVriesBengio:2015,MetivierBretaudeauBrossier:2014,YaoGholamiShen:2021} or inverse matrices in statistical applications \cite{ErikssonDongLee:2018,TangSaad:2011}. 
Other fixed-sparsity families of interest include block-diagonal matrices, banded matrices, and matrices whose non-zeroes correspond to a known graph (e.g., a low-dimensional mesh) \cite{ColemanMore:1983,HalikiasTownsend:2023,TangSaad:2011,VillaPetraGhattas:2021}. 

Any such family can be encoded by a binary matrix $\bv{S}\in \{0,1\}^{m\times n}$ representing the sparsity patter, where a one at position $(i,j)$ indicates that $(i,j)$ is allowed to be non-zero. The optimal Frobenius norm approximation to $\bv{A}$ with sparsity pattern $\bv{S}$ is easy to write down: it is given by $\bv{S}\circ \bv{A}$, where $\circ$ denotes the entrywise product. However, it is not typically possible to efficiently compute $\bv{S}\circ \bv{A}$ using a small number of matrix-vector products. In particular, while we can read any single column or row of $\bv{A}$ using one matrix-vector product (with a standard basis vector), $\bv{S}$'s ones could be distributed among many columns/rows.

As such, prior work has focused on \emph{approximation algorithms}. For the special case of diagonal approximation, it is known that the popular Hutchinson's estimator finds a near-optimal diagonal $\tilde{\bv{B}}$ satisfying \eqref{eq:goal} with $O(1/\epsilon)$ matvec queries, and no achieve 
algorithm achieving the same approximation guarantee can use asymptotically fewer queries \cite{BekasKokiopoulouSaad:2007,DharangutteMusco:2023}. 
Optimal query complexities are also known, e.g., for banded and block diagonal matrices.

For more complex sparse families, however, the optimal query complexity is often unknown. A number of prior results provide generic results that give query complexities depending on certain properties of $\bv{S}$. For example, if $\bv{S}$ has at most $s$ ones in any row (or column) then \cite{AmselChenKeles:2026b} shows that a $(1+\epsilon)$ error approximation can be learned using $O(s/\epsilon)$ products with random Gaussian vectors. It can also be shown that, if $\bv{S}$ has at most $\nnz(\bv{S})$ total ones, then $O(\sqrt{\nnz(\bv{S})/\epsilon})$ matvecs suffice, no matter how the ones are distributed \cite{AmselAviChen:2025,DasarathyShahBhaskar:2015}.

Another line of work studies algorithms based on graph coloring \cite{ColemanCai:1986,ColemanMore:1983,CurtisPowellReid:1974,GoldfarbToint:1984,Mccormick:1983,PowellToint:1979,TangSaad:2011}. Such methods are typically only analyzed in the easier \emph{recovery} setting, where $\bv{A}$ itself is assumed to be sparse. I.e., $\bv{S}\circ\bv{A} = \bv{A}$, and thus the optimal solution to \eqref{eq:goal} has $0$ error. In this setting, the complexity of these methods depends on the chromatic number (or related properties) of various graphs related to $\bv{S}$. For example, \cite{CurtisPowellReid:1974} observes that, if we can partition the columns of $\bv{S}$ into $k$ subsets, each consisting of columns with disjoint support, then any matrix with sparsity pattern $\bv{S}$ can be recovered with $k$ matvecs. The smallest such $k$ corresponds to the chromatic number of a ``conflict'' graph, $G$, associated with $\bv{S}$: $G$ has a vertex for each column and an edge between any two columns with overlapping support. 

Finding an optimal coloring of $G$ (to come up with a suitable set of query vectors) is NP-hard, so heuristic approximation methods are typically used instead. NP-hard optimization problems also arise in the implementation of more complex coloring-based algorithms with refined query complexities, which have been studied in work by Mccormick \cite{Mccormick:1983}, Power and Toint \cite{PowellToint:1979}, Coleman and Cai \cite{ColemanCai:1986}, and others. 

\subsection{Our Contributions}
\label{sec:contributions}
Our paper simplifies and generalizes this prior work by identifying a single combinatorial parameter of the sparsity pattern, $\bv{S}$, that, up to logarithmic factors, completely characterizes the matvec complexity of learning a near-optimal approximation with sparsity pattern $\bv{S}$. Specifically, generalizing a notion from graph theory, we define: 
\begin{definition}[Degeneracy of a Sparsity Pattern]
\label{def:degen}
Consider a matrix $\bv{S}\in \{0,1\}^{n\times m}$. For a positive integer $k$, delete all columns and rows of $\bv{S}$ with $\leq k$ ones, and then repeat this process on the smaller matrix that remains until no further columns or rows can be deleted. If we are left with an empty matrix, then we say $\bv{S}$ is $k$-degenerate. The \emph{degeneracy} of $\bv{S}$, denoted by $\degen(\bv{S})$, is the smallest value of $k$ such that $\bv{S}$ is $k$-degenerate.
\end{definition}
If $\bv{S}$ was the symmetric adjacency matrix of an undirected graph, $G$, then $\degen(\bv{S})$ would exactly correspond to the degeneracy of $G$, which is also known as the $k$-core number, width, or graph linkage \cite{KirousisThilikos:1996,LickWhite:1970}.\footnote{The degeneracy is also equivalent to $G$'s ``coloring number'' minus 1 (not to be confused with chromatic number) \cite{Eppstein:2025,ErdosHajnal:1966} and is within a 2 factor of $G$'s arboricity.} Our main algorithmic result is that $\degen(\bv{S})$ bounds the complexity of approximating an unknown matrix $\bv{A}$ with sparsity pattern $\bv{S}$:
\begin{theorem}[Upper Bound]
    \label{thm:main_upper}
Let $\bv{S}\in \{0,1\}^{n\times m}$ be a known sparsity pattern and $\bv{A}$ be an unknown matrix. There is a polynomial time algorithm that issues $O\left(\frac{\degen(\bv{S})}{\epsilon\delta} \cdot \log\left(n+m\right)\right)$ non-adaptive matrix-vector product queries with $\bv{A}$ and $\bv{A}^T$ and, with probability $1 - \delta$, returns a matrix $\tilde{\bv{B}}$ with sparsity pattern $\bv{S}$ satisfying:
\begin{align*}
\|\bv{A} - \tilde{\bv{B}}\|_F \leq (1+\epsilon)\cdot \min_{\bv{B}: \bv{B} = \bv{S}\circ\bv{B}}\|\bv{A} - \bv{B}\|_F = (1+\epsilon) \cdot \|\bv{A} - \bv{S}\circ\bv{A}\|_F.
\end{align*}
\end{theorem} 
Notably, unlike coloring methods, the set of queries issued by our algorithm can be computed in time polynomial in $n$ and $m$, as can the approximation $\tilde{\bv{B}}$. We also note that, in the easier ``recovery'' setting, where $\bv{S}\circ\bv{A} = \bv{A}$, our algorithm succeeds with exactly $2\cdot \degen(\bv{S})$ queries, and those queries can be chosen to simply be random Gaussian vectors.

We complement \Cref{thm:main_upper} with a  nearly matching lower bound that holds even for adaptive methods (\Cref{thm:main_upper} can be achieved with queries that are non-adaptively chosen): 
\begin{theorem}[Lower Bound]
    \label{thm:main_lower}
For any $\delta < 1$ and any finite $C$, no algorithm which queries less than $\degen(\bv{S})$ matvecs with $\bv{A}$ or $\bv{A}^T$ can return, with probability $1-\delta$, a matrix $\tilde{\bv{B}}$ satisfying $\|\bv{A} - \tilde{\bv{B}}\|_F \leq C\cdot \min_{\bv{B}: \bv{B} = \bv{S}\circ\bv{B}} \|\bv{A} - {\bv{B}}\|_F$.
\end{theorem} 
The proof of \Cref{thm:main_lower} is simple, following from the fact that any sparsity pattern $\bv{S}$ with degeneracy $\degen(\bv{S})$ must contain a submatrix with at least $\degen(\bv{S})$ ones in every row and column. It can be shown that, in the recovery setting where $\bv{A} = \bv{S}\circ\bv{A}$, recoverying even this submatrix in $\bv{A}$ requires $\geq \degen(\bv{S})$ matrix-vector product queries. In the recovery setting, $\min_{\bv{B}: \bv{B} = \bv{S}\circ\bv{B}} \|\bv{A} - {\bv{B}}\|_F = 0$, so a lower bound for recovery implies a lower bound for approximation for any finite $C$. 
This proof is detailed further in \Cref{sec:recovery}. 

\begin{remark}
    As mentioned, in the recovery setting, we can learn any $\bv{A}$
    with sparsity pattern $\bv{S}$ with \emph{exactly} $2\cdot \degen(\bv{S})$ matvecs (see \Cref{lem:recover_ub}). \Cref{thm:main_lower} thus implies that our recovery algorithm is essentially ``instance optimal'': no method can improve on its complexity by more than a factor of $2$ for \emph{any sparsity pattern}. At the same time, it avoids the computational intractability of, e.g., solving a graph coloring problem. We highlight some example sparsity patterns in \Cref{tab:complexity_measures} and \Cref{fig:sparsity_patterns} to compare our results to prior prior bounds.
\end{remark}

\begin{remark}
    For the approximation setting, there is a gap of $O\left(\log(m+n)/\epsilon\right)$ between the upper and lower bounds of \Cref{thm:main_upper} and \Cref{thm:main_lower}. We believe the $\log(m+n)$ term can possibly be eliminated from the upper bound. However, the $O(1/\epsilon)$ term is inherent. There are sparsity patterns for which $\Omega(\degen(\bv{S})/\epsilon)$ matvecs are necessary to learn a $(1+\epsilon)$-approximation. E.g., see \cite{AmselChenKeles:2026b} for a lower bound of $\Omega(s/\epsilon)$ for $s$-banded matrices.
    On the other hand, there are sparsity patterns where $O(\degen(\bv{S}))$ matvecs suffice. A simple example is when $\bv{S}$ has only a single non-zero column, $i$. Such a pattern had degeneracy $1$ and an exactly optimal approximation with the sparsity pattern (i.e., with $\epsilon = 0$) can be learned withe one matvec: simply read $\bv{A}$'s $i^\text{th}$ column by multiplying with a standard basis vector.
\end{remark}

\begin{table}[t]
\begin{center}
\begin{tabular}{l|c|c|c}
& \multicolumn{3}{c}{\textbf{Sparse Family}}\\
\textbf{Query Complexity Upper Bound} & {$s$-banded} & {$s$-modular} & {arrowhead}\\
\hline
Maximum Row Sparsity \cite{AmselChenKeles:2026b} & $O(s)$ & $O(s)$ & $n$ \\
$\sqrt{\text{Total Sparsity}}$ \cite{AmselAviChen:2025} & $O(\sqrt{ns})$ & $O(\sqrt{ns})$ & $O(\sqrt{n})$ \\
Conflict Graph Chromatic Number \cite{CurtisPowellReid:1974} & $O(s)$ & $O(s^2)$ & $n$\\
\hline
Degeneracy (this work) & $O(s)$ & $O(s)$ & $O(1)$\\
\end{tabular}
\end{center}
\caption{We give examples of natural sparse matrix families, along with various results on the query complexity of recovering matrices from these families. ``s-modular'' is a sparse family introduced in Section 4.3 of \cite{AmselChenKeles:2026b}. As we can see, our results improve on all prior work, and our \Cref{thm:main_lower} confirms that no further improvement is possible, for any family. Even for the harder approximation problem, any improvements are necessarily limited to the logarithmic factor and dependence on $\epsilon$.
}.
\label{tab:complexity_measures}
\end{table}

\begin{figure}
  \centering
  \tabcolsep=0.2\linewidth
  \divide\tabcolsep by 8
  \begin{tabular}{ccc}
     \begin{subfigure}[b]{0.2\textwidth}
        \includegraphics[width=\textwidth]{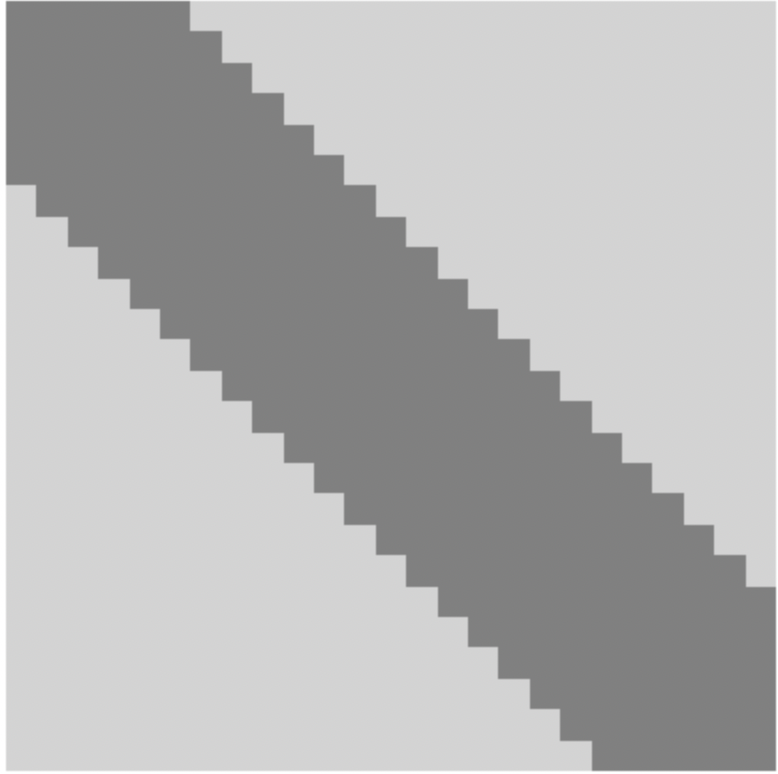}
        \end{subfigure} 
        \hspace{4em}
    \begin{subfigure}[b]{0.2\textwidth}
        \includegraphics[width=\textwidth]{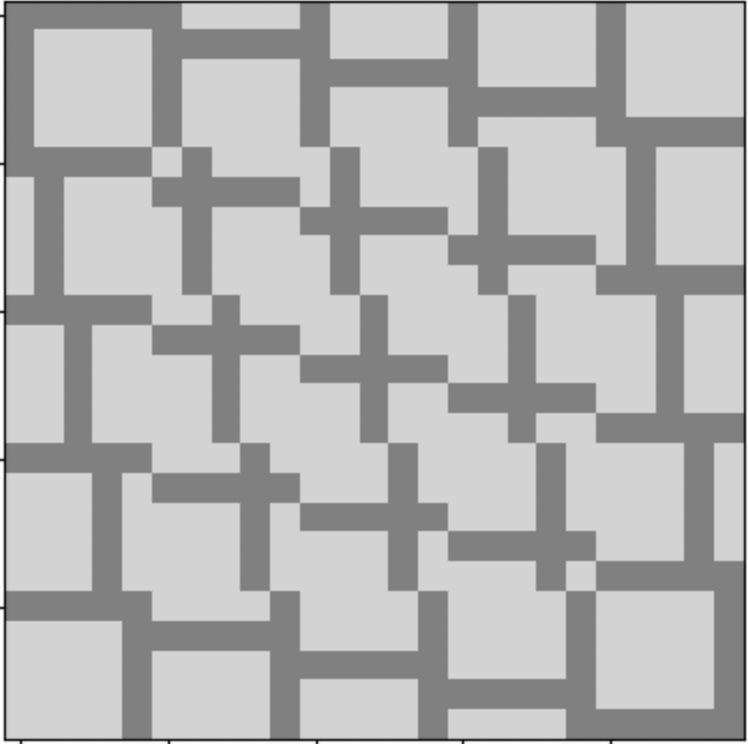}
    \end{subfigure} 
        \hspace{4em}
    \begin{subfigure}[b]{0.2\textwidth}
        \includegraphics[width=\textwidth]{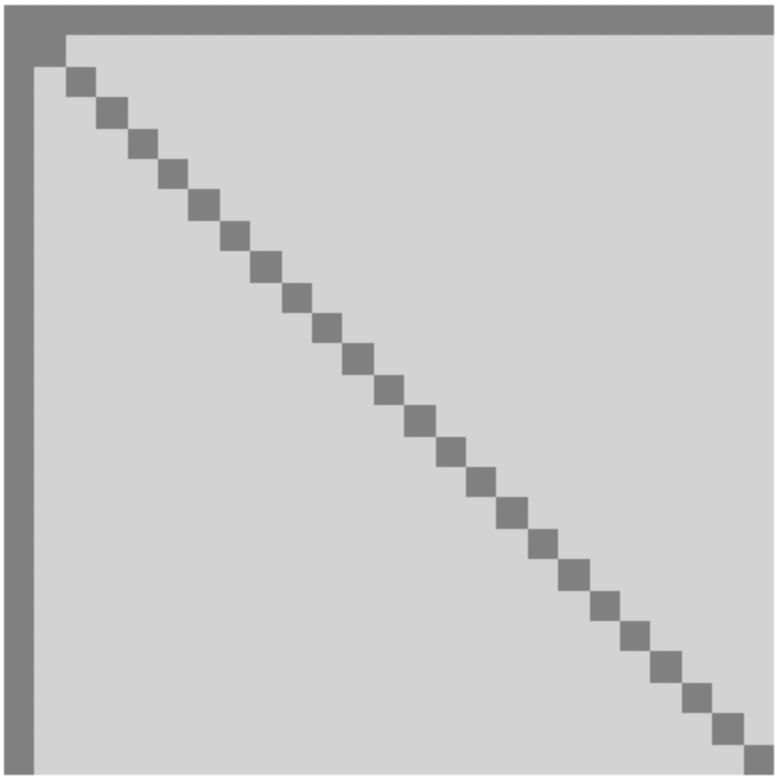}
    \end{subfigure} 
  \end{tabular}
  \caption{Illustrations of the $s$-banded, $s$-modular, and arrowhead sparsity patterns from \Cref{tab:complexity_measures} for 25x25 matrices with $s = 5$. Dark gray represents non-zeros, and light gray represents zeros.}
  \label{fig:sparsity_patterns}
\end{figure}

\subsection{Paper Structure}
We briefly review definitions and notation used throughout the paper in \Cref{sec:prelims}. In \Cref{sec:recovery}, we prove an upper bound of $\degen(\bv{S})$ matvecs and matching lower bound for recovering a matrix $\bv{A}$ with sparsity pattern $\bv{S}$. As discussed, this lower bound immediately implies \Cref{thm:main_lower}. In \Cref{sec:approximation}, we conclude by proving our main upper bound, \Cref{thm:main_upper}.

\section{Preliminaries}
\label{sec:prelims}

\subsection{Notation}
\label{sec:notation}
\noindent\textbf{Vectors, Matrices, and Norms.} We use bold upper case letters to denote matrices and bold lower case letters to denote vectors. For a vector $\bv{x}$, $x_i$ denotes the $i^\text{th}$ entry. 
For a matrix $\bv{M}$, $M_{i,j}$ denotes the $i,j$ entry. $\nnz(\bv{M})$ denote the number of non-zero entries in $\bv{M}$.
For matrices $\bv{A}$ and $\bv{B}$ of like dimensions, $\bv{A}\circ \bv{B}$ is the {entrywise product}. I.e., if $\bv{C} = \bv{A}\circ \bv{B}$, then $C_{i,j} = A_{i,j}\cdot B_{i,j}$.
For a length $n$ vector, $\|\bv{x}\|_2 = \left(\sum_{i=1}^n x_i^2\right)^{1/2}$ denotes the Euclidean norm and, for an $n\times m$ matrix, $\|\bv{M}\|_F = \left(\sum_{i=1}^n\sum_{j=1}^m M_{i,j}^2\right)^{1/2}$ denotes the Frobenius norm. We let $\bv{M}^+ \in \R^{m\times n}$ denote the Moore–Penrose pseudoinverse of a matrix $\bv{M}\in \R^{n\times m}$. When $\bv{M}$ has full column rank (as will always be the case in this paper), $\bv{M}^+ = (\bv{M}^T\bv{M})^{-1}\bv{M}^T$ and we have that $\bv{M}^+\bv{y} = \argmin_{\bv{x}\in \R^{m}} \|\bv{M}\bv{x} - \bv{y}\|_2$ for any $\bv{y}\in \R^{n}$. We let $\bv{0}$ denote the all zeros matrix. The dimension will be clear from context. 

\medskip
\noindent\textbf{Slice Notation.} We will also use matrix ``slice notation'' following standard conventions. For example, $\bv{M}_{i,:}$ denotes the $i^\text{th}$ row of $\bv{M}$ and $\bv{M}_{:,j}$ denotes the $j^\text{th}$ column. We always treat these vectors as column vectors (i.e., with width $1$). Suppose $\bv{M}$ has $n$ rows and $m$ columns and $\mathcal{S}$ and $\mathcal{R}$ are, respectively, ordered subsets of $\{1, \ldots, n\}$ and $\{1, \ldots, m\}$.  Then $\bv{M}_{\mathcal{S},\mathcal{R}}$ denotes the $|\mathcal{S}|\times |\mathcal{R}|$ matrix consisting of all entries $M_{i,j}$ for which $i \in \mathcal{S}, j\in \mathcal{R}$. Equivalently, let $\bv{s}\in \{0,1\}^n$ be a set indicator vector with $s_i = 1$ for $i \in \mathcal{S}$, $s_i = 0$ for $i\notin \mathcal{S}$, and let $\bv{r}$ be defined analogously for $\mathcal{R}$. Then we will also allow for $\bv{M}_{\mathcal{S},\mathcal{R}}$ to be written as $\bv{M}_{\bv{s},\bv{r}}$. $\bv{M}_{\mathcal{S},:} = \bv{M}_{\bv{s},:}$ denotes the matrix obtained by setting $\mathcal{R} = \{1, \ldots, m\}$ to be the complete set. For a binary vector, we overload notation and let $|\bv{s}|$ denote $|\mathcal{S}| = \nnz(\bv{s})$.

\medskip
\noindent\textbf{Index Restriction and Extension.}
Let $\bv{I}$ denote an $n\times n$ identity matrix and consider a set $\mathcal{S}\subseteq \{1,\ldots, n\}$, or equivalently, an indicator vector $\bv{s}\in \{0,1\}^n$ with $\bv{s}_i = 1$ for $i\in \mathcal{S}$. We will let $\bv{I}_{\mathcal{S}}$, or equivalently $\bv{I}_{\bv{s}}$, denote $\bv{I}_{\mathcal{S}} = \bv{I}_{\mathcal{S},:}$. Observe that, for any vector $\bv{x}\in \R^{n}$, $\bv{I}_{\mathcal{S}}\bv{x} = \bv{I}_{\bv{s}}\bv{x}$ performs a \emph{restriction operation}, returning a length $|\mathcal{S}|$ vector whose entries correspond to $x_i$ for $i\in \mathcal{S}$. Similarly, for any length $|\mathcal{S}|$ vector $\bv{y}$, $\bv{I}_{\mathcal{S}}^T \bv{y} = \bv{I}_{\bv{s}}^T \bv{y}$ performs an \emph{extension operation}, returning a length $n$ vector $\bv{x}$ with $\bv{x}_{s_i} = y_i$ for $\{s_1, \ldots, s_{|\mathcal{S}|}\} = \mathcal{S}$.
 
\subsection{Problem Formulation}
We study algorithms that can only access an unknown target matrix $\bv{A}\in \R^{n\times m}$ through left and right matrix-vector products (i.e., ``matvec queries'').  We formalize this model as follows:

\begin{definition}[Matrix-Vector Product Query Model]
An algorithm in this model can issue queries of the form $(\bv{x},s)$, where $\bv{x}$ is a vector and $s \in \{L,R\}$. If $s = R$ than the algorithm receives query response $\bv{A}\bv{x}$ for a fixed, unknown matrix $\bv{A}\in \R^{n\times m}$. Otherwise, if $s = L$, the algorithm receives $\bv{A}^T\bv{x}$. A $q$-query algorithm is one that issues $q$ queries, $(\bv{x}_1,s_1), \ldots, (\bv{x}_q,s_q)$, where the choice of $(\bv{x}_i,s_i)$ can depend on the outcome of all previous $i-1$ queries. 
\end{definition}
In the numerical analysis literature, methods implemented in the matvec query model are often referred to as ``matrix-free'' or ``implicit'' methods. In this work we consider solving two closely related problems in the matrix-vector product query model:
\begin{problem}[Sparse Matrix Recovery]
\label{prob:recovery}
Given a sparsity pattern $\bv{S}\in \{0,1\}^{n \times m}$ and query access to a matrix $\bv{A}$ where $\bv{A} \circ \bv{S} = \bv{A}$ (i.e., $\bv{A}$ has sparsity pattern $\bv{S}$), recover $\bv{A}$.
\end{problem}

\begin{problem}[Sparse Matrix Approximation]
\label{prob:approximation}
Given a sparsity pattern $\bv{S}\in \{0,1\}^{n \times m}$ and query access to an arbitrary matrix $\bv{A}$, find $\tilde{\bv{B}}$ such that $\tilde{\bv{B}} \circ \bv{S} = \tilde{\bv{B}}$ that satisfies:
\begin{align*}
\|\bv{A} - \tilde{\bv{B}}\|_F \leq (1+\epsilon) \min_{\bv{B}: \bv{B}\circ\bv{S} = \bv{B}}\|\bv{A} - {\bv{B}}\|_F.
\end{align*}
Note that we always have $\min_{\bv{B}: \bv{B}\circ\bv{S} = \bv{B}}\|\bv{A} - {\bv{B}}\|_F = \|\bv{A} - \bv{A}\circ \bv{S}\|_F$.
\end{problem}
Observe that \cref{prob:approximation} is strictly harder than \cref{prob:recovery}: if we have an algorithm that solves  \cref{prob:approximation} for any finite $\epsilon$, then when run on any matrix $\bv{A}$ with sparsity pattern $\bv{S}$, it must return $\bv{A}$ exactly to achieve multiplicative $(1+\epsilon)$ error (since $\min_{\bv{B}: \bv{B}\circ\bv{S} = \bv{B}}\|\bv{A} - {\bv{B}}\|_F = 0$). 

When evaluating algorithms for these problems in the matvec query model, we are primarily interested in minimizing matvec query complexity. This is motivated by the fact that, in many applications, computing $\bv{A}\bv{x}$ or $\bv{A}^T\bv{x}$ is a computational bottleneck, so query complexity closely correlates with computational cost. Indeed, there has been significant recent interest in proving tight upper and lower bounds on the matvec query complexity of natural linear algebraic problems like trace estimation, linear system solving, eigenvalue computation, and more \cite{BravermanHazanSimchowitz:2020,BravermanKrishnanMusco:2022,ChenTrogdonUbaru:2021,ChewiDios-PontLi:2024,DerezinskiEpperlyMeyer:2026,MeyerMuscoMusco:2021,MeyerSwartworthWoodruff:2025,MuscoMuscoRosenblatt:2025,SimchowitzEl-AlaouiRecht:2018,SunWoodruffYang:2019,WoodruffZhangZhang:2022}.

That said, other facets of an algorithm implemented in the matvec query model are also of interest. For example, ideally we should be able to efficiently compute the query vectors, $\bv{x}_1, \ldots, \bv{x}_q$, and we should also be able to efficiently construct $\tilde{\bv{B}}$ from the obtained responses. All of the algorithms in this paper run in time polynomial in the matrix dimensions $n$ and $m$.

There is also interest in algorithm that require limited ``adaptivity'', meaning that the queries can be chosen all at once (``non-adaptively'') or in a small number of rounds. Low adaptivity algorithms allow for increase parallelization query computation, which can lead to speedups in practice. All algorithms considered in this paper are fully non-adaptive. 


\section{Lower and Upper Bound for Sparse Matrix Recovery}
\label{sec:recovery}
We begin by observing that, for any sparsity pattern $\bv{S}$, $\degen(\bv{S})$ lower bounds the matvec complexity of recovering a matrix $\bv{A}$ with sparsity pattern $\bv{S}$. As discussed in \Cref{sec:contributions}, our main lower bound in \Cref{thm:main_lower} follows as an immediate corollary.
\begin{lemma}
    \label{lem:recovery_lb}
For any sparsity pattern $\bv{S}\in \{0,1\}^{n\times m}$, no algorithm can solve \Cref{prob:recovery} with any positive success probability using $< \degen(\bv{S})$ matvec queries.
\end{lemma}
\begin{proof}
    We first observe that, $\bv{S}$ must contain a submatrix with at least $\degen(\bv{S})$ ones in every row and column. This is because, if all submatrices had a row or column with $\leq  \degen(\bv{S})$ ones, then by definition, the matrix must have degeneracy $< \degen(\bv{S})$.

    Let $\bar{\bv{S}}\in \{0,1\}^{m'\times n'}$ denote this submatrix, where $m' \leq m$ and $n'\leq n$, and let $\ell' = \max(m',n')$. $\bar{\bv{S}}$ contains at least $\ell' \cdot \degen(\bv{S})$ ones. Now, consider the easier problem of recoverying a matrix $\bv{A}$ that is only non-zero on the non-zero entries of $\bar{\bv{S}}$. Every multiplication by $\bv{A}$ on the right by a vector $\bv{x}$ yields a vector with at most $m'$ non-zero entries, each of with is a linear combination of a subset of entries in $\bv{A}$ (those contained in a single row). Likewise, every multiplication by $\bv{A}$ on the left by a vector $\bv{x}$ yields a vector with at most $n'$ non-zero entries, each of which is a linear combination of a subset of entries in $\bv{A}$.

    Accordingly, no matter how our query vectors are chosen, if we issue $q$ matvec queries, we obtain \emph{at most} $q\cdot\ell'$ linear equations in the unknown entries of $\bv{A}$. It is impossible to determine these entries unless $q\cdot \ell' \geq \nnz(\bv{A}) \geq \ell' \cdot \degen(\bv{S})$ (since we have an underdetermined linear system \cite{HalikiasTownsend:2023}). In other words, we must have $q \geq \degen(\bv{S})$ to determine $\bv{A}$. 
\end{proof}

With \Cref{lem:recovery_lb} in place, it is natural to ask how close we can get to the limit of $\degen(\bv{S})$ matvec queries for solving \Cref{prob:recovery}. Surprisingly, we show that a simple, non-adaptive, polynomial time algorithm solves the problem with $2\cdot \degen(\bv{S})$ matvec queries. This is despite substantial work on more complex coloring-based algorithms for \Cref{prob:recovery}, which require solving NP-hard optimization problems to select query vectors, as discussed in \Cref{sec:prior_work}.

\begin{lemma}[Upper Bound]
\label{lem:recover_ub}
For any sparsity pattern $\bv{S}\in \{0,1\}^{n\times m}$, \Cref{prob:recovery} can be solved with probability 1 using $2\cdot  \degen(\bv{S})$ matvec queries (specifically, the product of each of $\bv{A}$ and $\bv{A}^T$ with $\degen(\bv{S})$ i.i.d. random Gaussian vectors).
\end{lemma}


\begin{proof}
    The algorithm that achieves \Cref{lem:recover_ub} is detailed in \Cref{alg:recovery}. The main idea is simple: we will iteratively recover rows and columns of $\bv{A}$ with $\leq \degen(\bv{S})$ unknowns.

\begin{algorithm}[ht]
\caption{Sparse Matrix Recovery}\label{alg:recovery}
    \begin{algorithmic}[1]
    \Require Sparsity pattern $\mathbf{S}\in \{0,1\}^{m\times n}$, matvec query access to a matrix $\mathbf{A} \in \R^{m\times n}$.
    \Ensure $\hat{\bv{A}}$ that is equal to $\bv{A}$ as long as $\bv{A}\circ\bv{S} = \bv{A}$ (i.e., $\bv{A}$ has sparsity pattern $\bv{S}$).
    	\algrule
    \State Compute $d = \degen(\bv{S})$. \Comment{This can be done in $O(\nnz(\bv{S}))$ time \cite{MatulaBeck:1983}.}
    \State Construct $\bv{W}\in \R^{m\times d}$ and $\bv{Z}\in \R^{n\times d}$ with i.i.d. random Gaussian entries.
    \State Compute $\bv{L} \gets \bv{A}^T\bv{W}$ and $\bv{R} \gets \bv{A}\bv{Z}$. \Comment{Requires $2\cdot \degen(\bv{S})$ matvec queries.}
    \State Set $\bv{L}_{init} \gets \bv{L}$ and $\bv{R}_{init} \gets \bv{R}$. Initialize $\hat{\bv{A}}$ as an $m\times n$ matrix of zeros. 
    \While{$\nnz(\bv{S}) > 0$}
        \State Let $\mathcal{R}$ contain the indices of all rows in $\bv{S}$ with $\leq d$ ones. 
        \State Let $\mathcal{C}$ contain the indices of all columns in $\bv{S}$ with $\leq d$ ones.
        \For{$r \in \mathcal{R}$}
            \State Set $\bv{s} \gets \bv{S}_{r,:}$ and set $\bv{G} = \bv{Z}_{\bv{s},:}^T$.
            \State Set $\hat{\bv{A}}_{r,:} \gets \bv{I}_{\bv{s}}^T\bv{G}^+ \bv{R}_{r,:}$. \Comment{$\bv{I}_{\bv{s}}^T$ is an $n\times |\bv{s}|$ extension operator (see \Cref{sec:prelims}).}
            \State Set all entries in row $r$ in $\bv{S}$ equal to $0$.
        \EndFor
        \For{$c \in \mathcal{C}$}
            \State Set $\bv{s} \gets \bv{S}_{:,c}$ and set $\bv{G} = \bv{W}_{\bv{s},:}^T$.
            \State Set $\hat{\bv{A}}_{:,c} \gets \bv{I}_{\bv{s}}^T\bv{G}^+ \bv{L}_{c,:}$.
            \State Set all entries in column $c$ in $\bv{S}$ equal to $0$.
        \EndFor
        \State Update $\bv{L} \gets \bv{L}_{init} - \hat{\bv{A}}^T\bv{W}$ and $\bv{R} \gets \bv{R}_{init} - \hat{\bv{A}}\bv{Z}$. \Comment{Does not require matvec queries.}
    \EndWhile
    \State Return $\hat{\bv{A}}$.
    \end{algorithmic}
\end{algorithm}
    
    Consider rows to start. Suppose we multiply $\bv{A}$ on the right by a random Gaussian matrix $\bv{Z}$ with $d = \degen(\bv{S})$ columns (using $\degen(\bv{S})$ matvec queries). Call the result $\bv{R} = \bv{A}\bv{Z}$. 
    Now, let $r$ be the index of any row in $\bv{S}$ with $\leq d$ ones. Let $\bv{s}$ denote this row. The $r^\text{th}$ row of $\bv{R}$ contains $d$ linear combinations of the entries in row $r$ of $\bv{A}$. Concretely, $\bv{R}_{r,:}^T = [\bv{A}_{r,\bv{s}}]^T\bv{Z}_{\bv{s},:}$. $\bv{Z}_{\bv{s},:}$ is a Gaussian matrix with $d$ columns and $|\bv{s}| \leq d$ rows, so it has full row-rank with probability $1$. We can thus solve for $\bv{A}_{r,\bv{s}}$ by computing $(\bv{Z}_{\bv{s},:}^T)^+\bv{R}_{r,:}$.\footnote{Recall that we take the convention that row and column slices are column vectors. So $\bv{R}_{r,:} \in \R^{d\times 1}$.}
    
    Symmetrically, we can learn the entries in all columns of $\bv{A}$ with $\leq d$ unknowns by multiplying $\bv{A}^T$ by a random Gaussian matrix ($\bv{W}$ in \Cref{alg:recovery}) with $d$ columns.

    Let $R$ and $C$ denote the indices of all rows and columns learned in this first round. We update $\bv{S}$ by zeroing out all entries in these rows and columns. Moreover, we update our current ``guess'' for $\bv{A}$, which is denoted by $\hat{\bv{A}}$, by setting the values in these rows and columns to those that we just learned. At this point, we can iterate on the matrix $\bv{A}-\bv{\hat{\bv{A}}}$ without computing additional matrix-vector products. We simply update $\bv{R}$ by subtracting off $\bv{\hat{A}}\bv{Z}$, which yields $\bv{R} = (\bv{A}-\bv{\hat{A}})\bv{Z}$. Similarly, we update $\bv{L} = \bv{A}^T\bv{W}$ by subtracting off $\bv{\hat{A}}^T\bv{W}$ to yield $\bv{L} = (\bv{A}-\bv{\hat{A}})^T\bv{W}$. From these matvecs, we proceed to learn all rows and columns in $\bv{A}-\bv{\hat{A}}$ with $\leq d$ unknowns. We repeat the process until $\bv{S}$ is all zeros and $\hat{\bv{A}} = \bv{A}$. 

    The correctness of the algorithm rests on the fact that all subsets of $\leq d$ rows from $\bv{Z}$ and $\bv{W}$ that we consider have full row rank. This holds simply because there are only a finite number of such subsets, and each has full row-rank almost surely \cite{Edelman:1988}. The runtime of the algorithm is also clearly polynomial in $n$ and $m$.
\end{proof}


\section{Near-Optimal Sparse Matrix Approximation}
\label{sec:approximation}
We conclude by providing an $\tilde{O}(\degen(\bv{S})/\epsilon)$ query algorithm for the harder problem of finding a near-optimal sparse approximation to an arbitrary matrix $\bv{A}$, which proves \Cref{thm:main_upper}. 
The basic approach is similar to the recovery algorithm from \Cref{sec:recovery}. Using standard tools from randomized numerical linear algebra (specifically, \Cref{lem:sketch}), it is not hard to show that multiplying $\bv{A}$ on the left and right by $O(\degen(\bv{S})/\epsilon)$ random Gaussian vectors suffices to near-optimally approximate any row and column with $\leq \degen(\bv{S})$ non-zeros. 

Once we have done so, a natural approach would be to subtract off our approximation to $\bv{A}$ and iterate on the residual. However, this leads to dependency issues: unlike the recovery setting, where we almost surely learn the exact values of rows and columns, in the approximation setting, our residual is a random matrix that depends on our random Gaussian queries. This prevents a direct argument along the lines of \Cref{lem:recover_ub}, since the same Gaussian queries cannot easily be ``reused'' on the residual matrix.

We deal with this dependency by simply drawing a \emph{fresh} set of random Gaussian vectors at every iteration of the algorithm. Naively, this could be a disastrous idea: it is possible to design sparsity patterns where the number of iterations scales with the matrix dimension. For example, let $\bv{S}$ be the adjacency matrix of a path graph on $n$ nodes. This pattern has degeneracy $1$, and if we iteratively peel off rows/columns with $1$ non-zero, it will take $n/2$ iterations. Our sample complexity will thus scale with $n$, instead of the desired $O(1/\epsilon)$. 

The fix is simple: it is not hard to prove that, if we peel off rows and columns with sparsity $\leq 4\cdot \degen(\bv{S})$ instead of $\leq \degen(\bv{S})$, then we can bound the number of iterations by $O(\log n + m)$. The only cost is a small increase in the number of Gaussian random vectors required per iteration (which scales with the sparsity of the rows and columns learned).

\begin{lemma}
    \label{lem:iteration_bound}
    Consider a sparsity pattern $\bv{S}\in \{0,1\}^{m\times n}$ with degeneracy $d = \degen(\bv{S})$. Suppose we delete all rows and columns in $\bv{S}$ with $\leq 4\cdot d$ ones, and then repeat this process on the remaining matrix. After $\log_2(m+n)$ rounds we are left with an empty matrix.
\end{lemma}
\begin{proof}
We first observe that any sparsity pattern with degeneracy $d$ has at most $d\cdot (m+n)$ ones. By an averaging argument, it follows that $< \frac{(m+n)}{4}$ columns have sparsity $> 4\cdot d$. Likewise, $< \frac{(m+n)}{4}$ rows have sparsity $> 4\cdot d$. Accordingly, if we remove all rows and columns with sparsity $\leq 4\cdot d$, we are left with a matrix with dimensions $m',n'$ such that:
\begin{align*}
m' + n' < \frac{(m+n)}{4} + \frac{(m+n)}{4} = \frac{m+n}{2}.
\end{align*}
In other words, the sum of $\bv{S}$' dimensions is reduced by a factor of at least two at every iteration. Hence, we obtain an empty matrix after $\log_2(m+n)$ iterations. 
\end{proof}

With \Cref{lem:iteration_bound} in place, we are ready to present and analyze our main algorithm. We require one additional fact on recovering a sparse approximation to a vector based on inner products with a set of Gaussian random vectors. Reminiscent of results in compressed sensing, this result follows from now standard techniques in randomized numerical linear algebra. See the analysis of Theorem 1 in \cite{AmselChenKeles:2026b} for a complete proof. 
\begin{lemma}
    \label{lem:sketch}
Let $\bv{y} \in \R^n$ be a vector and $\bv{s}\in \{0,1\}^n$ be a sparsity pattern of the same length. Let $\bv{H}\in \R^{q\times n}$ be a matrix with i.i.d. standard Gaussian entries.  Let:
\begin{align*}
    \hat{\bv{y}} = \bv{I}_\bv{s}^T\bv{H}_{:,\bv{s}}^+ \bv{H}\bv{y}, \text{ where $\bv{I}_\bv{s}^T$ is as defined in \Cref{sec:prelims}.}
\end{align*}
Then we have that: 
\begin{align}
    \label{eq:expected_error}
\E\left[\|\hat{\bv{y}} - \bv{s}\circ\bv{y}\|_2^2\right] 
        \leq \frac{|\bv{s}|}{q - |\bv{s}| - 1} \| \bv{y} - \bv{s}\circ \bv{y} \|_2^2.
\end{align}
\end{lemma}
\begin{remark}
\Cref{lem:sketch} bounds the difference between our sparse approximation, $\hat{\bv{y}}$, and the optimal approximation to $\bv{y}$ with sparsity pattern $\bv{s}$, which is $\bv{s}\circ \bv{y}$. Observe that, since $\bv{y} -\bv{s}\circ \bv{y}$ and $\bv{\hat{y}} - \bv{s}\circ \bv{y}$ have disjoint support, $\|\hat{\bv{y}} - \bv{s}\circ\bv{y}\|_2^2 + \| \bv{y} - \bv{s}\circ \bv{y} \|_2^2 = \|\bv{y} - \hat{\bv{y}}\|_2^2$. Accordingly, adding $\| \bv{y} - \bv{s}\circ \bv{y} \|_2^2$ to both sides of \eqref{eq:expected_error}, we can see that the bound is equivalent to
\begin{align*}
\E[\|\bv{y} - \hat{\bv{y}}\|_2^2] 
        \leq \left(1 + \frac{|\bv{s}|}{q - |\bv{s}| - 1}\right) \| \bv{y} - \bv{s}\circ \bv{y} \|_2^2.
\end{align*}
I.e., we can extract a $(1+\epsilon)$ near-optimal approximation to  $\bv{y}$ with sparsity pattern $\bv{s}$ from the result of $\bv{y}$'s multiplication with $q = \frac{|\bv{s}|}{\epsilon} + |\bv{s}| + 1 = O\left(\frac{|\bv{s}|}{\epsilon}\right)$ random Gaussian vectors.
\end{remark}
With \Cref{lem:sketch} in place, we proceed to our main result.
\begin{proof}[Proof of \Cref{thm:main_upper}]
    The result of \Cref{thm:main_upper} is obtain using \Cref{alg:sparse_approx} below.







\begin{algorithm}[ht]
\caption{Sparse Matrix Approximation}\label{alg:sparse_approx}
    \begin{algorithmic}[1]
    \Require Sparsity pattern $\mathbf{S}\in \{0,1\}^{m\times n}$, matvec query access to $\mathbf{A} \in \R^{m\times n}$, accuracy parameter $\epsilon \in (0,1)$, failure probability $\delta \in (0,1)$.
    \Ensure Matrix $\tilde{\mathbf B} = \tilde{\mathbf B}\circ \mathbf S$ with sparsity pattern $\bv{S}$ that approximates $\bv{A}$. 
    \algrule
    \State Compute $d \gets \degen(\bv{S})$. Set $q \gets \frac{4d}{\epsilon\delta/2} + 4d + 1$. \Comment{$q = O(d/\epsilon\delta)$.}
    \State Initialize $t \gets 1$ and row/column index sets $\mathcal{R}^{done} = \emptyset, \mathcal{C}^{done} = \emptyset$. 
   \While{$\nnz(\bv{S}) > 0$} \Comment{Preconstruct query vectors.}
        \State Set $\bv{S}^{t,1} \gets \bv{S}$. Let $\mathcal{R}^t$ contain the indices of all rows in $\bv{S}$ with $\leq 4d$ ones. 
        \State Construct $\bv{Z}^t \in \R^{n\times q}$ with i.i.d. Gaussian entries. Set $\bv{Z}^t_{\mathcal{C}^{done},:} \gets \bv{0}$. 
        \State Update $\mathcal{R}^{done} \gets \mathcal{R}^{done} \cup \mathcal{R}^t$, $\bv{S}_{\mathcal{R}^{t},:} \gets \bv{0}$
        \State Set $\bv{S}^{t,2} \gets \bv{S}$. Let $\mathcal{C}^t$ contain the indices of all columns in $\bv{S}$ with $\leq 4d$ ones. 
        \State Construct $\bv{W}^t \in \R^{m\times q}$ with i.i.d. Gaussian entries. Set $\bv{W}^t_{\mathcal{R}^{done},:} \gets \bv{0}$. 
        \State Update $\mathcal{C}^{done} \gets \mathcal{C}^{done} \cup \mathcal{C}^t$, $\bv{S}_{:,\mathcal{C}^{t}} \gets \bv{0}$.
        \State $t \gets t + 1$.
    \EndWhile  
    \State Compute $\bv{L}^i \gets \bv{A}^T\bv{W}^i$ and $\bv{R}^i \gets \bv{A}\bv{Z}^i$ for all $i=1,\ldots,t$. \Comment{Requires $2\cdot t\cdot q$ matvecs.} \label{line:query}

        \State Initialize $\tilde{\mathbf B}$ as an $m\times n$ all zeros matrix. 

        \For{$i=1,\ldots,t$}
            \For{$r \in \mathcal{R}^i$} \label{line:row_approx_loop}
                \State Set $\bv{s} \gets \bv{S}^{i,1}_{r,:}$ and set $\bv{H} = (\bv{Z}^i)^T$.\label{line:get_pattern_row}
                \State Set $\tilde{\bv{B}}_{r,:} \gets \bv{I}_\bv{s}^T\bv{H}_{:,\bv{s}}^+\bv{R}^i_{r,:}\,\,$. \label{line:row_approx}
            \EndFor
            \For{$c \in \mathcal{C}^i$}
                \State Set $\bv{s} \gets \bv{S}^{i,2}_{:,c}$ and set $\bv{H} = (\bv{W}^i)^T$.\label{line:get_pattern_column}
                \State Set $\tilde{\bv{B}}_{:,c} \gets \bv{I}_\bv{s}^T\bv{H}_{:,\bv{s}}^+\bv{L}^i_{c,:}\,\,$. \label{line:column_approx}
            \EndFor
        \EndFor
\State \Return $\tilde{\mathbf B}$
    \end{algorithmic}
\end{algorithm}
We first analyze the complexity of the algorithm. Each iteration of the while loop generates two query matrices with $q$ query vectors each. Moreover, by by \Cref{lem:iteration_bound}, the loop terminates after at most $t = \log_2(n + m)$ iterations. Hence, the total query complexity is $2\cdot q \cdot \log_2(n + m) = O\left(\frac{\degen(\bv{S})}{\epsilon\delta} \cdot \log(m + n)\right)$ matvec queries, as desired. Moreover, all queries are made non-adaptively: the query vectors are chosen by the while loop based on the structure of $\bv{S}$. All queries are then issued together on Line \ref{line:query}. Finally, it is clear that the algorithm runs in time polynomial in $n$ and $m$.

We next bound the approximation error. Note that $\mathcal{R}^i$ and $\mathcal{C}^i$ contain the indices of rows and columns that are processed at iteration $i$ of the main for loop. $\bv{S}^{i,1}$ is a sparsity pattern containing all entries of $\bv{S}$ that have \emph{not yet been learned} before iteration $i$. $\bv{S}^{i,2}$ is equal to $\bv{S}^{i,1}$ after further removing entries learned during iteration $i$ in the row for loop at Line \ref{line:row_approx_loop}.

We further define a mask matrix $\bv{M}^{i,1} \in \{0,1\}^{m\times n}$ so that $\bv{M}^{i,1}_{r,c} = 0$ if $c \in \mathcal{C}^1 \cup \ldots \cup \mathcal{C}^{i-1}$ and $\bv{M}^{i,1}_{r,c} = 1$ otherwise. Define $\bv{M}^{i,2} \in \{0,1\}^{m\times n}$ so that $\bv{M}^{i,1}_{r,c} = 0$ if $r \in \mathcal{R}^1 \cup \ldots \cup \mathcal{R}^{i-1}$ and $\bv{M}^{i,2}_{r,c} = 1$. 
Now, suppose we are at iteration $i$ of the main for loop. By constructing the query matrices $\bv{Z}^i$ and $\bv{W}^i$ with certain rows zero'd out, we can observe that:
\begin{align*}
\bv{R}^i &=  \bv{A}\bv{Z}^i =  (\bv{M}^{i,1} \circ \bv{A})\bar{\bv{Z}}^i & &\text{and} & \bv{L}^i &=  \bv{A}^T\bv{W}^i =  (\bv{M}^{i,2} \circ \bv{A})\bar{\bv{W}}^i, 
\end{align*}
where $\bar{\bv{Z}}^i \in \R^{n\times q}$ and $\bar{\bv{W}}^i \in \R^{m\times q}$ are Gaussian matrices with i.i.d. entries (equal to $\bv{Z}^i$ and $\bv{W}^i$ \emph{before} the rows in $\mathcal{C}^{done}$ and $\mathcal{R}^{done}$ where zero'd out).
Applying \Cref{lem:sketch} to the expression computed on Lines \ref{line:row_approx} and \ref{line:column_approx}, we thus have that, for all $r \in \mathcal{R}^i$ and $c \in \mathcal{C}^i$,
\begin{align}
    \label{eq:main101}
    \E\left[\|(\bv{S}^{i,1} \circ \tilde{\bv{B}})_{r,:} -  (\bv{S}^{i,1}\circ \bv{M}^{i,1} \circ \bv{A})_{r,:}\|_2^2\right] &\leq \frac{\epsilon\delta}{2}  \|(\bv{M}^{i,1} \circ \bv{A})_{r,:} - (\bv{S}^{i,1}\circ \bv{M}^{i,1} \circ \bv{A})_{r,:}\|_2^2,\\
    \E\left[\|(\bv{S}^{i,2} \circ \tilde{\bv{B}})_{:,c} -  (\bv{S}^{i,2}\circ \bv{M}^{i,2} \circ \bv{A})_{:,c}\|_2^2\right] &\leq \frac{\epsilon\delta}{2}  \|(\bv{M}^{i,2} \circ \bv{A})_{:,c} - (\bv{S}^{i,2}\circ \bv{M}^{i,2} \circ \bv{A})_{:,c}\|_2^2. \label{eq:main102}
\end{align}
Observe that the indices with value one in $\bv{S}^{i,1}$ are a subset of those with value one in  $\bv{M}^{i,1}$ since, like $\bv{M}^{i,1}$, $\bv{S}^{i,1}$ has every column in $\mathcal{C}^1 \cup \ldots \cup \mathcal{C}^{i-1}$ set to zero. Likewise, the indices with value one in $\bv{S}^{i,2}$ are a subset of those $\bv{M}^{i,2}$. This allows us to simplify the right hand sides of \eqref{eq:main101} and \eqref{eq:main102} by removing the $\bv{M}^{i,1}$ and $\bv{M}^{i,2}$ terms. Additionally, referring to the right hand sides, we observe that $\bv{S}^{i,1}_{r,:} = \bv{M}^{i,1}_{r,:}\cdot \bv{S}_{r,:}$ and $\bv{S}^{i,2}_{:,2} = \bv{M}^{i,2}_{:,c}\cdot \bv{S}_{:,c}$. We thus have:
\begin{align*}
(\bv{M}^{i,1} \circ \bv{A})_{r,:} - (\bv{S}^{i,1}\circ \bv{M}^{i,1} \circ \bv{A})_{r,:} &= (\bv{M}^{i,1}\circ \bv{A})_{r,:} - (\bv{M}^{i,1}\circ\bv{S}\circ \bv{A})_{r,:} = \bv{A}_{r,:} - (\bv{S}\circ \bv{A})_{r,:} \\
(\bv{M}^{i,2} \circ \bv{A})_{:,c} - (\bv{S}^{i,2}\circ \bv{M}^{i,2} \circ \bv{A})_{:,c} &= (\bv{M}^{i,2}\circ \bv{A})_{:,c} - (\bv{M}^{i,2}\circ  \bv{S}\circ \bv{A})_{:,c} = \bv{A}_{:,c} - (\bv{S}\circ \bv{A})_{:,c}.
\end{align*}
Plugging into \eqref{eq:main101} and \eqref{eq:main102}, we obtain
\begin{align}
    \label{eq:main10}
    \E\left[\|(\bv{S}^{i,1} \circ \tilde{\bv{B}})_{r,:} -  (\bv{S}^{i,1}\circ \bv{A})_{r,:}\|_2^2\right] &\leq \frac{\epsilon\delta}{2}  \|\bv{A}_{r,:} - (\bv{S} \circ \bv{A})_{r,:}\|_2^2,\\
    \E\left[\|(\bv{S}^{i,2} \circ \tilde{\bv{B}})_{:,c} -  (\bv{S}^{i,2}\circ \bv{A})_{:,c}\|_2^2\right] &\leq \frac{\epsilon\delta}{2}  \|\bv{A}_{:,c} - (\bv{S}\circ \bv{A})_{:,c}\|_2^2.
\end{align}
Finally, we bound the total approximation error, $\|\bv{\tilde{B}} - \bv{S}\circ \bv{A}\|_F^2$, by noticing that every non-zero entry in $\bv{\tilde{B}} - \bv{S}\circ \bv{A}$ appears either as a non-zero in $(\bv{S}^{i,1} \circ \tilde{\bv{B}})_{r,:} -  (\bv{S}^{i,1}\circ \bv{A})_{r,:}$ or in $(\bv{S}^{i,2} \circ \tilde{\bv{B}})_{:,c} -  (\bv{S}^{i,2}\circ \bv{A})_{:,c}$ for some value of $i$. We thus have:
\begin{align*}
\E\left[\|\bv{\tilde{B}} - \bv{S}\circ \bv{A}\|_F^2\right] &\leq \frac{\epsilon\delta}{2}\sum_{r=1}^m \|\bv{A}_{r,:} - \bv{S}_{r,:}\circ \bv{A}_{r,:}\|_2^2 + \frac{\epsilon\delta}{2}\sum_{c=1}^n \|\bv{A}_{:,c} - \bv{S}_{:,c}\circ \bv{A}_{:,c}\|_2^2 \\
&\leq \frac{\epsilon\delta}{2}\|\bv{A} - \bv{S}\circ \bv{A}\|_F^2 + \frac{\epsilon\delta}{2}\|\bv{A} - \bv{S}\circ \bv{A}\|_F^2 = \epsilon\|\bv{A} - \bv{S}\circ \bv{A}\|_F^2.
\end{align*} 
By Markov's inequality, it follows that $\|\bv{\tilde{B}} - \bv{S}\circ \bv{A}\|_F^2 \leq \epsilon \|\bv{A} - \bv{S}\circ \bv{A}\|_F^2$ with probability at least $1-\delta$. Finally, since $\bv{\tilde{B}} - \bv{S}\circ \bv{A}$ and $\bv{A} - \bv{S}\circ \bv{A}$ have disjoint support, we can add $\|\bv{A} - \bv{S}\circ \bv{A}\|_F^2$ to both sides of the equation to conclude that, with probability at least $1-\delta$,
\begin{align*}
\|\bv{\tilde{B}} - \bv{S}\circ \bv{A}\|_F^2 + \|\bv{A} - \bv{S}\circ \bv{A}\|_F^2  = \|\bv{A} - \bv{\tilde{B}}\|_F^2  \leq (1+\epsilon)\cdot \|\bv{A} - \bv{S}\circ \bv{A}\|_F^2.
\end{align*}
We conclude by noting that the dependence on $1/\delta$ in \Cref{thm:main_upper} could be improved to logarithmic using standard techniques, like the high-dimensional median trick (see, e.g., \cite{AmselChenKeles:2026b} for details). We omit this improvement for simplicity of presentation.
\end{proof}

\section{Conclusion and Future Directions}
Our work nearly resolves the complexity of fixed-sparsity pattern matrix approximation (and recovery) in the matvec query model, showing that the degeneracy of the sparsity pattern characterizes query complexity. Some questions remain open, however. For example, can we remove the $\log(m+n)$ factor in \Cref{thm:main_upper} so that it matches the lower bound in \Cref{thm:main_lower}? Additionally, is it possible to implement \Cref{alg:sparse_approx} using only ``unstructured'' query vectors (e.g., random Gaussians) instead of the ``masked'' random vectors we currently use?

It is also interesting to ask if our generic results for sparse families can be extended to other classes of matrix families. 
For example, sparse families are a special case of \emph{linearly parameterized} matrix families. A linearly parameterized family consists of all matrices that can be written as $\sum_{i=1}^q c_i B_i$ for a chosen set of coefficients $c_1, \ldots, c_q \in \R$ and fixed basis matrices $B_1, \ldots, B_q \in \R^{m\times n}$ \cite{HalikiasTownsend:2023}. It is know that, in the recovery setting, such matrices can always be learned with $\sqrt{q}$ matvecs \cite{AmselAviChen:2025}. However, many linear families (including sparse families) require fewer, often on the order of $q/n$. Is there a simple parameter of $\{B_1, \ldots, B_q\}$ that governs the query complexity? What about general, non-linear matrix families?

\section{Acknowledgements}
We would like to thank Cameron Musco for helpful discussions. CM is partially supported by NSF Award \#2427363.


\bibliography{refs.bib}

\end{document}